\documentclass[runningheads]{llncs}
\usepackage[T1]{fontenc}
\usepackage[utf8]{inputenc}
\usepackage{amssymb,amsmath}
\usepackage{color}
\usepackage{cite}
\usepackage{hyperref}
\usepackage{enumerate}
\usepackage[font=small,labelfont=bf]{caption}
\usepackage[subrefformat=parens,labelfont=default,font=small]{subcaption}

\usepackage{todonotes}

\title{Line and Plane Cover Numbers Revisited\thanks{\lncsarxiv{The
      full version of this article is available at Arxiv
      \cite{arxiv}.}{}  S.F.\ was supported by DFG grant FE$\,$340/11-1,
    A.W.\ by DFG grant WO$\,$758/9-1, and T.B.\ by NSERC.}}
    
\author{Therese Biedl\inst{1} 
\and Stefan Felsner\inst{2}\orcidID{0000-0002-6150-1998}
\and Henk~Meijer \inst{3}
\and Alexander~Wolff\inst{4}\orcidID{0000-0001-5872-718X}
}

\institute{University of
  Waterloo, Waterloo, Canada
  \and TU Berlin, Berlin, Germany
  \and University College Roosevelt, The Netherlands 
  \and Universit\"at W\"urzburg, W\"urzburg, Germany
}
  
\authorrunning{T.~Biedl et al.}

\newtheorem{thm}{Theorem}

\DeclareMathOperator{\tn}{tn}

\newcommand{\lncsarxiv}[2]{#2}

\graphicspath{{figures/}}

\let\doendproof\endproof
\renewcommand\endproof{~\hfill$\qed$\doendproof}

\begin{document}

\maketitle

\begin{abstract}
  A measure for the visual complexity of a straight-line crossing-free
  drawing of a graph is the
  minimum number of lines needed to cover all vertices.  For a given
    graph $G$, the minimum such number (over all drawings in  dimension
  $d \in \{2,3\}$) is called the \emph{$d$-dimensional weak line
    cover number} and denoted by $\pi^1_d(G)$.
  In 3D, the minimum number of \emph{planes} needed to cover all
  vertices of~$G$ is denoted by $\pi^2_3(G)$.  When edges are also
  required to be covered, the corresponding numbers $\rho^1_d(G)$ and
  $\rho^2_3(G)$ are called the \emph{(strong) line cover number} and
  the \emph{(strong) plane cover number}.

  Computing any of these cover numbers---except $\pi^1_2(G)$---is
  known to be NP-hard.  The complexity of computing $\pi^1_2(G)$
  was posed as an open problem by Chaplick et al.\ [WADS 2017].
  We show that it is NP-hard to decide, for a given planar
  graph~$G$, whether $\pi^1_2(G)=2$.  We further 
  % show that, for every planar 3-tree~$G$, $\pi^1_2(G)=O(\log n)$,
  % which is asymptotically tight.
  show that the universal stacked triangulation of depth~$d$, $G_d$,
  has $\pi^1_2(G_d)=d+1$.
  Concerning~3D, we show that any $n$-vertex graph~$G$ with
  $\rho^2_3(G)=2$ has at most $5n-19$ edges, which is tight.
\end{abstract}

\section{Introduction}

Recently, there has been considerable interest in
representing graphs with as few objects as possible.
The idea behind this objective is to keep the
visual complexity of a drawing low for the observer.
The types of objects that have been used are straight-line segments
\cite{desw-dpgfs-CGTA07, dm-dptfs-CCCG14, hkms-dpgfg-JGAA18,
  kms-eaadfs-JGAA18} and circular arcs \cite{s-dgfa-JGAA15,
  hkms-dpgfg-JGAA18}.

Chaplick et al.~\cite{cflrvw-dgflfp-GD16} considered \emph{covering}
straight-line drawings of graphs by 
%unbounded objects (lines, planes)
lines or planes
and defined the following new graph parameters.
Let $1 \leq l < d$, and let $G$ be a graph. The \emph{$l$-dimensional affine cover}
number of~$G$ in~$\mathbb R^d$, denoted by $\rho^l_d(G)$, is defined as the minimum
number of $l$-dimensional planes in~$\mathbb R^d$ such that~$G$ has a
crossing-free straight-line drawing that is contained in the union of
these planes.  The \emph{weak} $l$-dimensional affine cover
number of~$G$ in~$\mathbb R^d$, denoted by $\pi^l_d(G)$, is defined similarly
to $\rho^l_d(G)$, but under the weaker restriction that only the vertices
%(and not necessarily the edges) of~$G$ 
are contained in the union of the planes.
Clearly,  $\pi^l_d(G)\le\rho^l_d(G)$, and
if $l'\le l$ and $d'\le d$ then~$\pi^l_d(G)\le
\pi^{l'}_{d'}(G)$ and $\rho^l_d(G)\le \rho^{l'}_{d'}(G)$.  
It turns out that it suffices to study the parameters~$\rho^1_2$, $\rho^1_3$,
$\rho^2_3$, and $\pi^1_2$, $\pi^1_3$, $\pi^2_3$:

\begin{thm}[Collapse of the Affine Hierarchy \cite{cflrvw-dgflfp-GD16}]
  For any integers $1\le l<3\le d$ and for any graph $G$, it holds
  that $\pi^l_d(G)=\pi^l_3(G)$ and $\rho^l_d(G)=\rho^l_3(G)$.
\end{thm}

%% Reinsert if space
% For a given graph~$G$, $\pi^1_2(G)\leq \rho^1_2(G)$ and it can be much
% smaller.  For instance, Chaplick et al.\ showed that for the
% nested-triangles graph $T_k=C_3\times P_k$ with $n=3k$ vertices it
% holds that $\rho^1_2(T_k)\ge n/2$, whereas clearly $\pi^1_2(T_k)\le3$.

Disproving a conjecture of Firman et al.~\cite{flsw-ewlc2-GD18},
Eppstein~\cite{Epp19} constructed planar, cubic, 3-connected,
bipartite graphs on $n$ vertices with $\pi^1_2(G) \geq n^{1/3}$.
Answering a question of Chaplick et al.~\cite{cflrvw-dgflfp-GD16}
%concerning the $\pi^1_2$-value of graphs of treewidth~2, Eppstein
he also
constructed a family of subcubic series-parallel graphs with unbounded
$\pi^1_2$-value.  Felsner~\cite{Fe19} proved that, for every
4-connected plane triangulation~$G$ on $n$ vertices, it holds that
$\pi^1_2(G) \leq \sqrt{2n}$.
Chaplick et al.~\cite{cflrvw-cdgfl-WADS17} also investigated the
complexity of computing the affine cover numbers.  Among others, they
showed that in 3D, for $l \in \{1,2\}$, it is NP-complete to decide
whether $\pi^l_3 (G) \leq 2$ for a given graph~$G$.  In 2D, the
question has still been open, but a related question was raised by
Dujmovi\'c et al.~\cite{dpw-tlg-DMTCS04} already in 2004.  They
investigated so-called \emph{track layouts} which are defined as
follows.  A graph admits a $k$-track layout if its vertices can be
partitioned into $k$ ordered independent subsets such that any pair of
subsets induces a plane graph (w.r.t.\ the order of the subsets).  The
track number of a graph~$G$, $\tn(G)$, is the smallest~$k$ such that
$G$ admits a $k$-track layout.  See also \cite{DJMMUW} for some recent
developments.
Note that in general
$\pi^1_2(G) \ne \tn(G)$; for example, $\pi^1_2(K_4)=2$, whereas
$\tn(K_4)=4$.  Note further that a 3-track layout is necessarily plane
(which is not the case for $k$-track layouts with $k>3$).  Dujmovi\'c
posed the computational complexity of $k$-track layout as an open
question.

While it is easy to decide efficiently whether a graph admits a
2-track layout, Bannister et al.~\cite{bddew-tllpd-Algorithmica19}
answered the open question of Dujmovi\'c et al.\ already for 3-track
layouts in the affirmative.  They first showed that a graph has a
leveled planar drawing if and only if it is bipartite and has a
3-track layout.  Combining this results with the NP-hardness of level
planarity, proven by Heath and Rosenberg~\cite{hr-loguq-SICOMP92},
immediately showed that it is NP-hard to decide whether a given graph
has a 3-track layout.  For $k>3$, deciding the existence of a
$k$-track layout is NP-hard, too, since it suffices to add to the
given graph $k-3$ new vertices each of which is incident to all
original vertices of the graph~\cite{bddew-tllpd-Algorithmica19}.

\paragraph{Our contribution.} 
We investigate several problems concerning the {\em weak line cover number} $\pi_2^1(G)$
and the {\em strong plane cover number} $\rho_3^2(G)$:
\begin{itemize}
\item We settle the open question of Chaplick et
  al.~\cite[p.~268]{cflrvw-cdgfl-WADS17} 
%  concerning the computational
%  complexity of the weak line cover number~$\pi^1_2$ 
  by showing that
  it is NP-hard to test whether, for a given planar graph~$G$,
  $\pi^1_2(G)=2$; see Section~\ref{sec:NP}.
\item We show 
  % that any planar 3-tree $G$ has $\pi^1_2(G)=O(\log n)$, and 
  that $G_d$, the universal stacked triangulation of depth~$d$, (which
  has treewidth~3) has $\pi^1_2(G_d)=d+1=\log_3(2n_d-5)+1$, where
  $n_d$ is the number of vertices of~$G_d$; see
  Section~\ref{sec:threetrees}.
\item Eppstein has identified classes of treewidth-2 graphs with
  unbounded $\pi^1_2$-value.  We give an easy direct argument showing
  that some 2-tree $H_d$ with $n'_d$ vertices has $\pi^1_2(H_d)\in
  \Omega(\log n'_d)$; see \lncsarxiv{the full
    version~\cite{arxiv}}{Appendix~\ref{sec:2trees}}.
\item Concerning 3D, we show that any $n$-vertex graph~$G$ with
  $\rho^2_3(G)=2$ has at most $5n-19$ edges; see
  Section~\ref{sec:twoplanes}.  This bound is tight.
\end{itemize}

\section{Complexity of Computing Weak Line Covers in 2D}
\label{sec:NP}

In this section we investigate the computational complexity of
deciding whether a graph can be drawn on two lines.  
%Note that any
%graph that can be drawn on two parallel lines can also be drawn on two
%intersecting lines.  Furthermore, any graph that can be drawn on two
%intersecting lines can be drawn on the two coordinate axes.  Hence,
%these are the two lines that we will assume from now on, with the
%origin being their intersection point.

\begin{thm}
  It is NP-hard to decide whether a given plane (or planar) graph~$G$
  admits a drawing with $\pi^1_2(G)=2$.
\end{thm}

\begin{proof}
  Our proof is by reduction from the problem \textsc{Level Planarity},
  which Heath and Rosenberg~\cite{hr-loguq-SICOMP92} proved to be
  NP-hard.  The problem is defined as follows.  A planar graph~$G$ is
  \emph{leveled-planar} if its vertex set can be partitioned into sets
  $V_1,\dots,V_m$ such that $G$ has a planar straight-line drawing
  where, for every $i \in \{1,\dots,m\}$, vertices in~$V_i$ lie on the
  vertical line $\ell_i \colon y=i$ and each edge~$v_j v_k$ of~$G$
  connects two vertices on consecutive lines (that is, $|j-k|=1$).

  % We exploit that Heath and Rosenberg showed that \textsc{Level
  %   Planarity} is NP-hard even if the embedding of the given graph is
  % fixed and if we insist that the outer face of the given plane
  % graph~$G$ is \emph{level-monotone}, that is, the outer face is a
  % cycle~$C$ that can be split into two paths starting and ending in
  % the same pair of vertices such that the indices of the levels
  % increase monotonically along both paths.
  
  Chaplick et al.~\cite{cflrvw-dgflfp-GD16} have shown that every
  leveled-planar graph can be drawn on
  two lines.  The converse, however, is not true.  For example, $K_4$
  is not leveled-planar, but $\pi^1_2(K_4)=2$.  
  %%% reinsert if space
  % We need to ensure that
  % (i)~edges are not drawn on the lines (but ``between'' them),
  % (ii)~the drawing does not use the origin, and (iii)~there are enough
  % levels (line segments) that are separated from each other.
  % 
  Therefore, we modify the given graph in three ways.  (a)~We replace
  each edge of~$G$ by a $K_{2,4}$-gadget where the two nodes in one
  set of the bipartition replace the endpoints of the former edge; see
  Fig.~\ref{fig:K24-gadget}.  (b)~We add to the graph~$G'$ that
  resulted from the previous step a new subgraph~$G_0$ (two copies
  of~$K_4$ sharing exactly two vertices), which we
  connect by a path to a vertex on the outer face of~$G$.  (If the
  outer face is not fixed, we can try each vertex.)  In
  Fig.~\ref{fig:spiral}, $G_0$ is yellow % (light gray)
  and the path is red.  The
  length~$L$ of the path is any upper bound on the number of levels
  of~$G'$, e.g., the diameter of~$G'$ (plus~1).  (c)~We attach
  to~$G_0$ a triangulated spiral~$S$ (dark green in
  Fig.~\ref{fig:spiral}).  The spiral makes $L+2$ right turns; its
  final vertex is identified with the outermost vertex of the previous
  turn.  Hence, apart from its many triangular faces, the graph
  $S+G_0$ has a large inner face~$F$ of degree $2(L+2)$ and a quadrangular
  outer face.  Let~$G''$ be the resulting graph.
  It remains to show that $G$ is leveled-planar if and only if $\pi_2^1(G'')=2$.

  ``$\Rightarrow$'': Fix a 
  leveled-planar drawing of~$G$.  
  % reinsert the following if space
%  We first show that the graph~$G'$ is
%  leveled-planar.  To this end, we introduce new levels
%  $\ell'_1,\dots,\ell'_{m-1}$ such that, for $i\in\{1,\dots,m-1\}$,
%  level~$\ell'_1$ lies between levels~$\ell_i$ and~$\ell_{i+1}$; see
%  Fig.~\ref{fig:K24-gadget}.  For each original edge~$uv$ with
%  $u \in \ell_i$ and~$v \in \ell_{i+1}$, we place the four vertices
%  of~$K_{2,4}$ on the new level~$\ell'_i$ and connect them to~$u$
%  and~$v$.  This yields the leveled-planarity of~$G'$.
By doubling the layers and using the new layers to place the large sides of $K_{2,4}$'s, one easily sees that $G'$ is also leveled-planar, see Fig.~\ref{fig:K24-gadget}.
  As shown in Fig.~\ref{fig:spiral}, the large inner face~$F$
  of~$S+G_0$ can be drawn so that it partitions the half{}lines emanating from the origin
  into~$L$ levels.  (It is no problem that consecutive levels are
  turned by $90^\circ$.)  Since we chose~$L$ large enough (in
  particular $L \ge 2m-1$), we can easily draw~$G'$ inside~$F$.  Note
  that the red path attached to~$G_0$ is long enough to reach any
  vertex on the outer face of~$G'$.  Hence, $\pi^1_2(G'')=2$.
  
  \begin{figure}[tb]
    \begin{subfigure}[b]{.44\textwidth}
      \centering
      \includegraphics{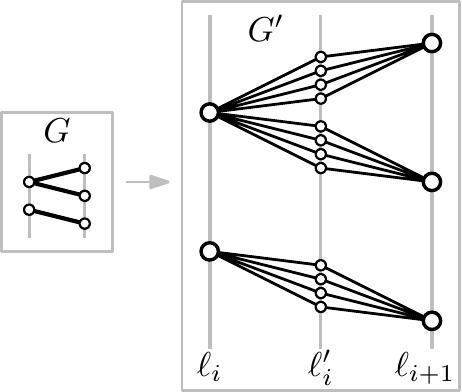}
      \caption{transforming~$G$ to~$G'$}
      \label{fig:K24-gadget}
    \end{subfigure}
    \hfill
    \begin{subfigure}[b]{.5\textwidth}
      \centering
      \includegraphics{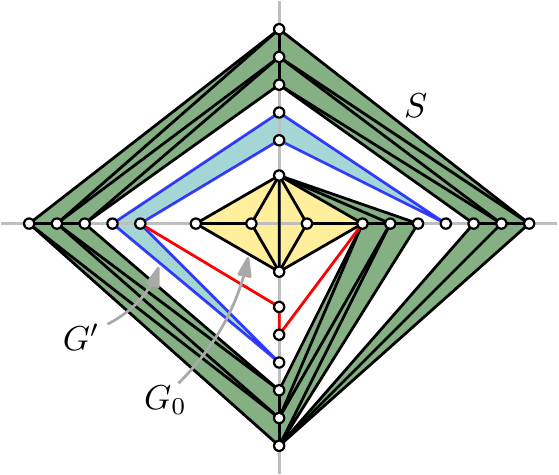}
      \caption{transforming~$G'$ to~$G''$}
      \label{fig:spiral}
    \end{subfigure}
    \caption{Our reduction from \textsc{Level Planarity}}
    \label{fig:reduction}
  \end{figure}

  ``$\Leftarrow$'': Fix a drawing of~$G''$ on two lines.  The two
  lines cannot be parallel since $G''$ contains $K_{2,4}$ and is not
  outer-planar; so after translation and/or skew we may assume that
  these two lines are the two coordinate axes.
  It is not hard to verify that~$G_0$ must be drawn such that the
  origin is in its interior, at the common edge of the two $K_4$'s.
  Furthermore, given this drawing of~$G_0$, the 3-connected spiral~$S$
  must be drawn as in Fig.~\ref{fig:spiral}.  Due to
  planarity and the fact that~$G'$ is connected to~$G_0$ via the red
  path, $G'$ can only be drawn in the interior of~$F$.  The drawing of
  $S+G_0$ partitions the half{}lines emanating from the origin into
  levels, which we number $1,2,\dots$ starting from the innermost
  level that contains a vertex of~$G'$.  Inside this face, the only
  way to draw the $K_{2,4}$-gadgets is as in
  Fig.~\ref{fig:K24-gadget}, spanning three consecutive levels.  This
  forces all vertices of~$G$ to be placed on the odd-numbered levels
  and the vertices in $G'-G$ on the even-numbered levels.  Now we can
  get a level assignment for~$G$ by reverting the transformation in
  Fig.~\ref{fig:K24-gadget}.
 % : we put all vertices of~$G$ lying on level
 % $2k+1$ in the drawing of~$G'$ into a new level~$k$.  Note that only
 % the order of the vertices in each level matters.
  Hence, $G$ is leveled-planar.

  % This finishes the proof that our reduction is correct.  It is
  % obvious that the reduction can be performed in time polynomial in
  % the size of~$G$.
  This shows that our reduction is correct.  It runs in polynomial
  time.
\end{proof}

\section{Weak Line Covers of Planar 3-Trees in 2D}
\label{sec:threetrees}

In this section we consider the weak line cover number~$\pi^1_2$ for
planar graphs,
i.e., we are interested in crossing-free straight-line drawings 
with vertices located on a small collection of lines.
Clearly $\pi^1_2(G)=1$ if and only if $G$ is a forest of paths.
The set of graphs with $\pi^1_2(G)=2$, however, is already
surprisingly rich, it contains all trees, outerplanar graphs and
subgraphs of grids~\cite{bddew-tllpd-Algorithmica19,Epp19}.

\emph{Stacked triangulations}, a.k.a.~planar 3-trees or Apollonian
networks, are obtained from a triangle by repeatedly selecting a
triangular face~$T$ and adding a new vertex (the {\em vertex stacked
  inside $T$}) inside~$T$ with edges to the vertices of~$T$.  This
subdivides~$T$ into three smaller triangles, the \emph{children}
of~$T$.
% By extracting a planar subgraph $G'$ of pathwidth 3 (which has
% $\pi_2^1(G')\leq 3$ \cite{BCDM}) and recursing on the subgraphs
% (which are show the following (see the appendix for details).

% \begin{thm}
% $\pi_2^1(G)\in O(\log n)$ for any 3-tree $G$.
% \end{thm}
For $d\geq 0$ let~$G_d$ be the \emph{universal stacked
  triangulation of depth $d$}, defined as follows. 
The graph~$G_0$ is a triangle~$T_0$, and~$G_d$ (for $d\geq 1$) is
obtained from~$G_{d-1}$ by adding a
stack vertex in each bounded face of~$G_{d-1}$.
Graph~$G_d$ has $n_d=\frac{1}{2}(3^d+5)$
vertices and $3^d$ bounded faces.  We show that its weak line cover
number is $d+1=\log_3(2n_d-5)+1 \in \Theta(\log n_d)$.
(A lower bound of $d$
can also be found in Eppstein's recent book~\cite[Thm. 16.13]{EppBook}.)

\begin{thm}\label{thm:stacked}
  For $d\geq 1$ it holds that $\pi(G_d) = d+1$.
\end{thm}

\begin{proof}
Here we prove only the lower bound; the construction for the upper
bound is illustrated in Fig.~\ref{fig:stacked} and given in 
\lncsarxiv{the full version}{Appendix~\ref{app:stacked}}.
Let $\cal{L}$ be a family of lines covering the vertices of a drawing
of~$G_d$.  Let~$a$, $b$, and~$c$ be the vertices of~$T_0$.  We first argue
that at least $d$ lines are needed to cover $V\setminus T_0$.
Let~$x_1$ be stacked into~$T_0$.  There is a line
$L_1\in \cal{L}$ covering $x_1$.  Note that $L_1$ can intersect only
two of the three child triangles of~$T_0$
(where ``intersect'' here means ``in the interior''). Let~$T_1$ be a
child triangle avoided by~$L_1$, and let $x_2$ be the vertex stacked
into~$T_1$.  There is a line $L_2\in \cal{L}$ covering $x_2$.
Let~$T_2$ be a child triangle of~$T_1$ avoided by~$L_2$. Iterating this
%we extract a sequence $L_1,L_2,\ldots L_d$ of 
yields $d$ pairwise distinct lines in~$\cal{L}$.

To find one additional line in~$\cal{L}$, we distinguish some cases.
If a line $L\in \cal{L}$ covers two vertices of~$T_0$, then it covers
no inner vertex, and we are done.

Assume some line $L_a\in \cal{L}$ intersects~$x_1$ and one vertex
of~$T_0$, say $a$.  Let~$L_b$ and~$L_c$ be the lines intersecting~$b$
and~$c$.  The lines~$L_a$, $L_b$, and~$L_c$ are pairwise different,
else we are in the previous case.  Of the three child triangles of
$T_0$, at most one is intersected by $L_a$ and at most two each are
intersected by $L_b$ and~$L_c$.  Therefore, some child triangle~$T_1$
of~$T_0$ is intersected by at most one of~$L_a$, $L_b$, or~$L_c$.  The
graph~$G_{d-1}$ inside~$T_1$ requires at least $d-1$ lines for its
interior points, and at most one of those lines is~$L_a$, $L_b$,
or~$L_c$, so in total at least $d+1$ lines are needed.

The argument is similar if no line covers two of $a$, $b$, $c$,
and~$x_1$.  The four distinct lines supporting $a$, $b$, $c$,
and~$x_1$ then intersect at most two child triangles each.  So one
child triangle~$T_1$ is intersected by at most two of these lines.
Combining the $d-1$ lines needed for the interior of~$T_1$ with the
two lines that do not intersect it, shows that $d+1$ lines are
needed.
\end{proof}

\begin{figure}[tb]
  \centering
  \includegraphics[scale=.18,trim=800 0 0 0,clip]{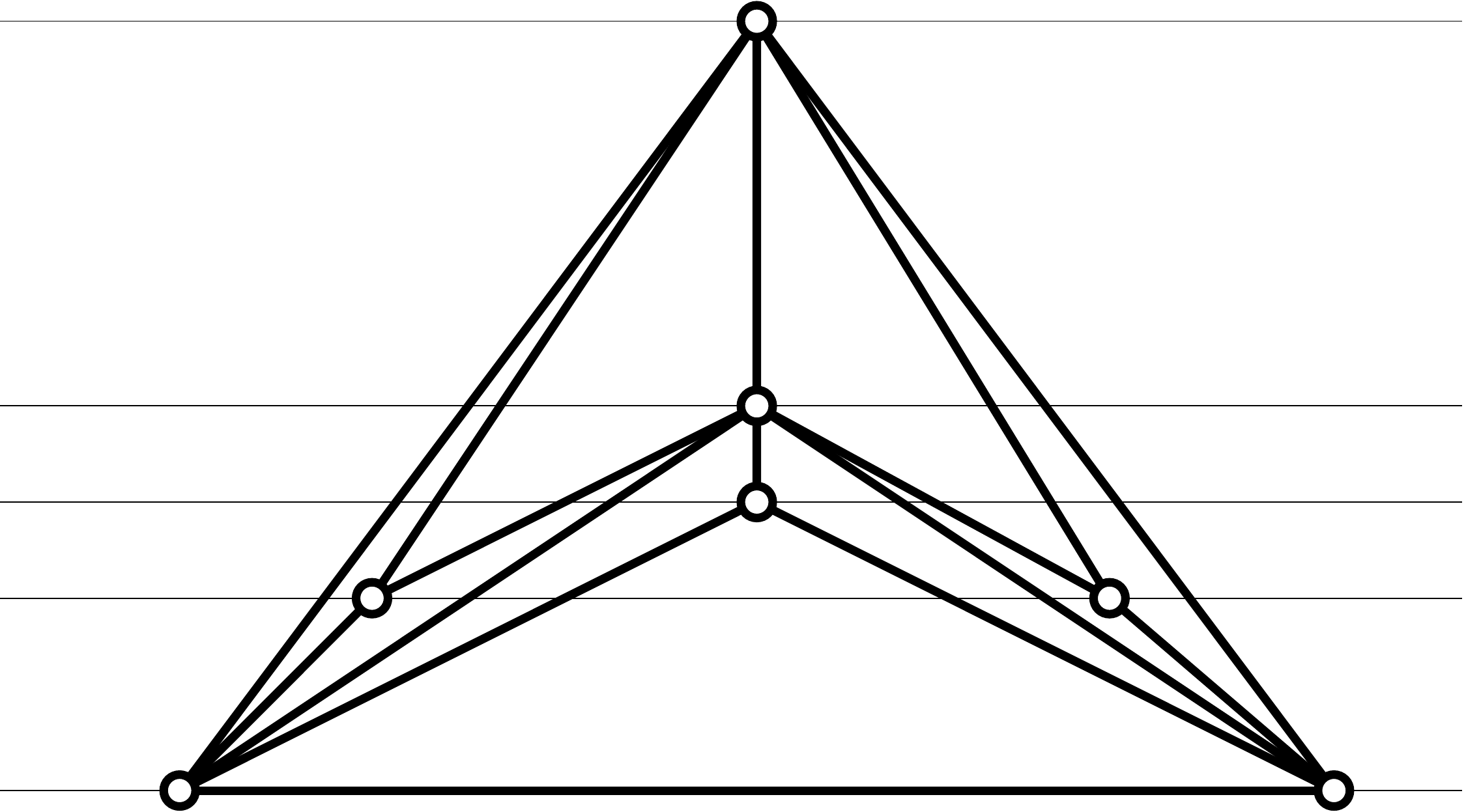}
  \caption{A drawing of $G_2$ that can be extended to a drawing of
    $G_3$ on 5 parallel lines.}
  \label{fig:stacked}
\end{figure}

\section{Maximal Graphs on Two Planes in 3D}
\label{sec:twoplanes}

We now switch to dimension $d=3$ and the strong cover number.
Obviously any graph $G$ with a drawing that is covered by two planes
has at most $6n-12$ edges since it is the union of two planar graphs.
Using maximality arguments and % some 
counting, we show that in fact $G$
has at most $5n-19$ edges if $n\geq 7$.
(The restriction $n \ge 7$ is required since for
$n=3,4,5,6$ we can have $3,6,9,12$ edges.)  

We argue first that our bound is tight.
The {\em spine} is the intersection of two planes $A$ and $B$. 
Put a path with $n-4$ vertices on the spine.  Add
one vertex in each of the four halfplanes and connect each of these
vertices to all vertices on the spine and to the vertex on the
opposite halfplane\lncsarxiv{. (We provide a figure in the full
  version.)}{; see Fig.~\ref{fig:tightness} in
  Appendix~\ref{app:boundarycases}.}
This yields $n-5$ edges on the path and $2(n-4)+1$ edges in each of
the two planes, so $5n-19$ edges in total.

\begin{thm}
  \label{thm:twoplanes}
  Any graph~$G$ with $\rho^2_3(G)=2$ and $n\geq 7$ vertices has at
  most $5n-19$ edges.
\end{thm}
\begin{proof}
Fix a drawing of~$G$ on planes $A$ and $B$, inducing planar
graphs~$G_A$ and~$G_B$ within those planes.  Let $G_A^+$ and $G_B^+$
be the graphs obtained 
from~$G_A$ and~$G_B$ by adding any edge that can be inserted without
crossing, within the same plane, and with at most one bend on the spine.
Clearly it suffices to argue that $G_A^+$ and $G_B^+$ together have at
most $5n-19$ edges.
Let $s$ be the number of vertices on the spine, let $a$ be the number
of vertices of~$G_A^+$ not on the spine, and let $b$ be the
number of vertices of~$G_B^+$ not on the spine.  Clearly, $a+b+s=n$.
We may assume $a\leq b$.   
We also assume that $1\leq s\leq n-4$ and that at least one edge
of~$G_A^+$ crosses the spine (so $2\leq a\leq b$); see
\lncsarxiv{the full version}{Appendix~\ref{app:boundarycases}}.

Let $t$ be the number of edges drawn along the spine.  These are the
only edges that belong to $G_A^+$ and $G_B^+$.  Since $G_A^+$ and $G_B^+$
have at least three vertices each, we can bound the number of edges
of~$G$, $m(G)$, as follows:
\begin{align}
\label{equ:1}
m(G) &\le m(G_A^+) + m(G_B^+)-t \,\le\, 3(s+a)-6 + 3(s+b)-6 - t \\
&= 3n-12 + 3s - t \,\le\, 4n-16+2s-t. \nonumber
\end{align}
So we must show that $2s-t\leq n-3$.  Let an {\em internal gap} be a
line segment connecting two consecutive, non-adjacent vertices on the spine.
There are $s-t-1$ internal gaps.  Let the {\em external gap} be the two
infinite parts of the spine.  Note that at least one edge of $G_A^+$ must
cross the external gap, because $G_A^+$ has at least one vertex on each side
of the gap, and we could connect the extreme such vertices (or re-route an existing
edge) to cross the external gap, perhaps using a bend on the spine.
We may further assume that even after such re-routing every internal gap is crossed by at
least one edge of~$G_A^+$.  Otherwise we could delete all edges of~$G_B^+$
passing through the gap, insert the edge between the spine vertices, and
re-triangulate the drawing of $G_B^+$ where we removed edges.  This
would remove an internal gap, but would not decrease the number of edges.
Since no edge can cross two gaps, at least $s-t$ edges of $G_A^+$ cross gaps.
These edges form a planar bipartite graph with
at most $a$ vertices; therefore $s-t\leq 2a-3$.%
\footnote{One might be tempted to
write a bound of $2a-4$ here, but we must allow for the possibility of
$a=2$, in case of which the planar bipartite graph may have $1=2a-3$ edges.}  
This yields
$2s-t \leq s + 2a-3 \leq s + a + b - 3 = n-3$
as desired.
\end{proof}

\noindent
We conjecture that the following more general statement
holds:

\begin{quote}
  Any $n$-vertex graph $G$ with $\rho^2_3(G)=k$ has at most
  $(2k+1)(n-2k)+k-1$ edges, for all large enough $n$.
\end{quote}

% Roadmap for a proof:
% We need many vertices participating in the planar graphs on many of the
% $k$ planes. This is best achieved when all the planes share a common spine.
% If the arrangement of the planes is established, the proof of the case
% with two planes can be adapted.

\paragraph{Acknowledgments.}
This research started at the Bertinoro Workshop on Graph Drawing
2017. We thank the organizers and other participants, % for discussions,
in particular Will Evans, Sylvain Lazard, Pavel Valtr, Sue Whitesides,
and Steve Wismath.  We also thank Alex Pilz and Piotr Micek for
enlightening conversations.

\bibliographystyle{splncs04} %{plainurl}
\bibliography{abbrv,pi}

\lncsarxiv{\end{document}}{}

\newpage
\appendix
\section*{Appendix: Missing Proofs}

\section{Rest of the Proof of Theorem~\ref{thm:stacked}}
\label{app:stacked}

For the upper bound, we draw
$G_k$ (for $k=0,\dots,d$) on $d+2$ layers, i.e., distinct
horizontal lines.  For every bounded face of~$G_k$, one edge
is \emph{short} (i.e., either horizontal or connecting two adjacent layers)
while the other two edges each cross at least $d-k$ layers in their interior.
See also Fig.~\ref{fig:stacked}.  For $G_0$, do this by placing
$(b_0,c_0)$ horizontally on the lowest layer and
$a_0$ on the highest layer.

Assume $G_k$ (for $0\leq k<d$) has been drawn in this way, and consider a
bounded face $T_k=\{a_k,b_k,c_k\}$ of $G_k$ into which we want to place the
stacked vertex $x_k$ to get a drawing of $G_{k+1}$.  
Say $(b_k,c_k)$ is the short edge.  Hence the two edges incident to $a_k$
cross at least $d-k\geq 1$ layers in their interior.    Place $x_k$ on
the layer adjacent to $a_k$ and interior to $T_k$ and verify all conditions.
For $k=d$ we hence get a drawing of $G_d$ on $d+2$ layers.  Observe that the
top two layers contain only $a_0$ and the vertex $x_0$ stacked inside $T_0$.
(This exists by $d\geq 1$.)
Hence the line through $a_0,x_0$, together with the $d$ lines through the other
$d$ layers, gives a set of $d+1$ lines supporting the drawing.

\section{Weak Line Covers of 2-Trees}
\label{sec:2trees}

We already had the operation of stacking a vertex inside a triangle.
We now introduce a similar operation, \emph{stacking a vertex onto
an edge $(a,b)$}, which consists of adding a new vertex $x$ adjacent
to $a$ and $b$.  Define $H_0$ to be the graph consisting of a single
edge $(a,b)$, and let $H_d$ (for $d\geq 1$) be the graph obtained
from $H_{d-1}$ by stacking a vertex onto every edge of $H_{d-1}$.
The graph~$H_d$ has $3^d$ edges and (since it is a 2-tree)
$n'_d=\frac{3^d+3}{2}$ vertices.

\begin{theorem}
  \label{thm:twotrees}
  $\pi_2^1(H_d) \geq 1+\lfloor \frac{d}{8} \rfloor \in \Omega(\log n'_d)$.
\end{theorem}

\begin{proof}
  Fix an arbitrary straight-line planar drawing and line cover of
  $H_{d+8}$ (for some $d\geq 0$); we will show that this line cover
  needs at least one line more than a line cover of~$H_d$.  The
  theorem then holds by induction since $H_0$ needs one line.

  Let $H_0=\{(a,b)\}$ be the original edge from which~$H_{d+8}$ was
  built.  Let $v_1,\dots,v_5$ be the common neighbours of~$a$ and~$b$
  acquired as we extended~$H_0$ to~$H_5$ and hence stacked onto
  $(a,b)$ five times.  Let $L$ be the line in the line cover
  of~$H_{d+8}$ that supports~$a$.  By the pidgeon-hole principle, at
  least three of $v_1,\dots,v_5$ must lie in one (closed)
  half-space~$h$ of~$L$; say, $v_1$, $v_2$, and~$v_3$.  Sort them such
  that the rotation at $a$ contains (in counterclockwise order) a ray
  along $L$, $(a,v_1)$, $(a,v_2)$, $(a,v_3)$, the other ray along $L$
  (with the first pair and last pair possibly coinciding).  See
  Fig.~\ref{fig:twotrees}.

\begin{figure}[tb]
\hspace*{\fill}
\includegraphics[page=1]{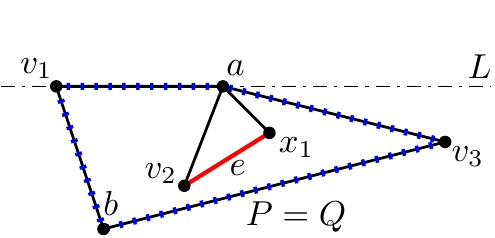}
\hspace*{\fill}
\includegraphics[page=2]{twotree}
\hspace*{\fill}
\caption{Finding an edge $e$ (thick red) that is inside polygon $P$ (dotted blue).}
\label{fig:twotrees}
\end{figure}

Let $Q$ be the quadrilateral $\langle a,v_1,b,v_3\rangle$.  Observe
that its sides are edges of~$H_5$; hence, they have no crossings.
When extending~$H_5$ to~$H_8$, we stack onto edge $(a,v_2)$ three
times; say with vertices~$x_1$, $x_2$, and~$x_3$.  We now distinguish
two cases depending on the location of~$b$:

\medskip\noindent{\bf Case 1: } $b$ is also in $h$.
Then $Q$ lies entirely within $h$, and its angle at $a$ is convex or flat.
In particular, edge $(a,v_2)$ (which lies between $(a,v_1)$ and $(a,v_3)$)
enters the interior of $Q$.  By planarity it crosses no edge of $Q$, so
$v_2$ (and with it also $x_1$) lie strictly inside $Q$.  Set $e=(v_2,x_1)$
and $P=Q$.

\medskip\noindent{\bf Case 2: } $b$ is not in $h$.
Then the angle of~$Q$ at~$a$ is reflex or flat.  This implies that $v_2$ (and
with it $x_1,x_2,x_3$) lie outside $Q$.  Therefore the edges $(a,x_i)$
for $i \in \{1,2,3\}$ must lie between $v_1$ and $v_3$ in the rotation at~$x$;
say the rotation is $v_1,x_1,x_2,x_3,v_3$ (with $v_2$ somewhere inbetween).
Since we have a straight-line drawing, $x_1$, $x_2$, and~$x_3$ lie
in~$h$, too.  Let $Q'$ be the quadrilateral
$\langle a,x_1,v_2,x_3 \rangle$.  With the same argument as in the previous case
(but using $v_2$ in place of~$b$), we see that $x_2$ lies strictly within~$Q'$.
Set $e=(v_2,x_2)$ and $P=Q'$.

\medskip
In both cases we have found a polygon $P$ such that $L$ does not intersect
its interior, and an edge $e$  that
lies strictly inside $P$ except perhaps at an endpoint (but that endpoint
is not on $L$).  Edge
$e$ has a graph $H_d$ stacked onto it, and none of the vertices of this $H_d$ 
(which are either ends of $e$ or strictly inside $P$) can be supported by $L$.
Hence, a line cover of~$H_{d+8}$ must contain at least one line more
than a line cover of~$H_d$.
\end{proof}

\section{Missing Cases for the Proof of Theorem~\ref{thm:twoplanes}}
\label{app:boundarycases}

Now we consider the boundary cases.
\begin{itemize}
\item If $s=0$ or $a=0$ or $a+b\le 2$ then $G$ is planar and $m(G) \le
  3n-6 < 5n-19$ (since $n \ge 7$).  Therefore we may assume $a\ge 1$,
  $a+b\ge 3$, and $s=n-a-b \le n-3$.
\item If $s=n-3$ then $a=1$ and $G$ consists of a planar graph in
  plane~$B$ on $n-1$ vertices plus a unique other vertex in~$A
  \setminus B$ adjacent to at most $s$ vertices on the spine.
  Therefore $m(G) \le 3(n-1)-6+s \le 4n-12 \le 5n-19$ by $n\geq 7$.
  So we may assume $s \le n-4$, hence $b \ge 3$.
\item Assume now that all vertices of $G_A^+$ are to one side of the
  spine or on the spine.  Observe that we may assume $s\geq 3$, for if
  $s\leq 2$ then, by Equation~\ref{equ:1} (which did not use that the
  vertices of $G_A^+$ occur on both sides), we have $m(G) \le 3n-12+3s
  \le 3n-6 < 5n-19$.

  Since $s\geq 3$, the convex hull of the drawing of $G_A^+$ contains
  at least $s+1\geq 4$ vertices, hence $m(G_A^+) \le 3(s+a)-7$.  This
  strengthens Equation~\ref{equ:1} to $m(G) \le 4n-17+2s-t$, so it
  suffices to show $2s-t\leq n-2$.  We can therefore afford to have no
  edge in the external gap.  There are no internal gaps (because those
  could be filled with edges with the same argument as before), so
  $s-t=1$ and $2s-t=s+1\leq n-2$ as desired.
\end{itemize}

\begin{figure}[tb]
  \centering
  \includegraphics[scale=.91]{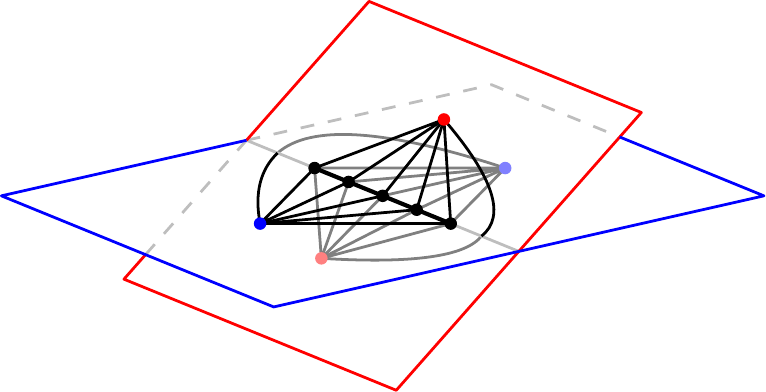}
  \caption{Example that shows that the bound $5n-19$ for the number of
    edges of a 2-plane graph is tight.}
  \label{fig:tightness}
\end{figure}
 
\end{document}